\RequirePackage{fix-cm}
\documentclass[smallextended]{svjour3}       
\smartqed  

\usepackage{graphicx}

\usepackage{amsmath}
\usepackage{amssymb, hyperref}
\usepackage[mathscr]{euscript}
\usepackage{cancel}
\usepackage[colorinlistoftodos]{todonotes}
\usepackage{color, ulem}

\usepackage{longtable, multirow, verbatim}

\usepackage{natbib}

\newtheorem{conj}{Conjecture}

\DeclareMathOperator{\spn}{span}
\DeclareMathOperator{\rank}{rank}
\newcommand{\mbf}[1]{\mathbf{#1}}

\def\bgamma{{\boldsymbol \gamma}}
\begin{document}

\title{Gain and Loss of Function mutations in biological chemical reaction networks: a mathematical model with application to colorectal cancer cells
}



\author{Sara Sommariva        \and
        Giacomo Caviglia \and Michele Piana 
}


\institute{S. Sommariva \at
              Dipartimento di Matematica, Universit\'a di Genova, via Dodecaneso 35 16146 Genova, Italy \\
              Tel.: +39-010-3536644\\
              \email{sommariva@dima.unige.it}           
           \and
              G. Caviglia \at
              Dipartimento di Matematica, Universit\'a di Genova, via Dodecaneso 35 16146 Genova, Italy \\
              Tel.: +39-010-3536830\\
              \email{caviglia@dima.unige.it} 
              \and
               M. Piana \at
              Dipartimento di Matematica, Universit\'a di Genova, and CNR - SPIN GENOVA, via Dodecaneso 35 16146 Genova, Italy \\
              Tel.: +39-010-3536939\\
              \email{piana@dima.unige.it} 
}

\date{}

\maketitle

\begin{abstract}
This paper studies a system of Ordinary Differential Equations modeling a chemical reaction network and derives from it a simulation tool mimicking Loss of Function and Gain of Function mutations found in cancer cells. More specifically, from a theoretical perspective, our approach focuses on the determination of moiety conservation laws for the system and their relation with the corresponding stoichiometric surfaces. Then we show that Loss of Function mutations can be implemented in the model via modification of the initial conditions in the system, while Gain of Function mutations can be implemented by eliminating specific reactions. Finally, the model is utilized to examine in detail the G1-S phase of a colorectal cancer cell. 
\end{abstract}


\section{Introduction} \label{sec:introd}

Signalling networks are Chemical Reaction Networks (CRNs) consisting of an interconnected set of pathways, modeling the flow of chemical reactions initiated by information sensed from the environment through families of receptor ligands \citep{Jordan,Sever,Tyson}. The reactions in the network follow the process of information transfer inside cytosol down to the description of the activity of a few related transcription factors \citep{Karin,kohn_2006}. Similarly to what happens for many networks of biological interest, a  signaling network may comprise hundreds of reacting chemical species and reactions. 

Any CRN can be mapped into a system of ordinary differential equations (ODEs) by following standard procedures \citep{Feinberg,Yu,Chellaboina} thus allowing for simulations of the kinetics of the signaling process in the biologic case (see, e.g., \cite{Anderson,Roy}). 
In principle, the resulting mathematical model is capable of describing healthy physiologic states, and can be adapted to simulate individual pathological conditions associated to mutations. Since most cancer diseases result from an accumulation of a set of mutations \citep{Stratton,Levine}, the numerical solution of this mathematical model represents a convenient computational tool for the simulation of the mechanisms giving rise to a tumor. Further, this kind of models can typically be tuned in order to mimic the effects of targeted therapies \citep{Fachetti,Levine_T} and drug resistance mechanisms \citep{Eduati}.

The present paper first realizes an analysis of the system of ODEs associated to a CRN, characterized by a level of generality sufficient to provide an efficient simulation tool. From a formal viewpoint, we apply mathematical methods devised for the investigation of deterministic homogeneous chemical reaction systems based on mass-action kinetics. However, in our approach a crucial role is played by moiety conservation laws (CLs), which are essential in the determination and interpretation of results \citep{De_Martino, Shinar}. Further, a geometric classification of equilibrium states in terms of stoichiometric surfaces is discussed, together with their stability properties. 

From a more operating viewpoint, we define projection operators, which map the system of ODE for physiologic conditions into the mutated system which models Loss of Function (LoF) and Gain of Function (GoF) mutations \citep{Griffiths,Li}. The model is built and made operative in such a way that it can be modified almost straightforwardly by addition or elimination of chemical reactions or chemical species, change in the values of the rate constants in the formulation of mass-action laws, or change of the initial conditions. As a consequence, the parameters values that have been originally considered as fixed can be customized to fit the tumor data of a specific patient.

As an application, we examine in detail the kinetics of a recently proposed network which simulates how colorectal cancer (CRC) cells process information from external growth factors and the related answer \citep{Tortolina}. The network is focused on the G1-S transition point. In this transition a newborn cell in the G1 phase of its cycle must pass the control of this checkpoint before starting the S phase of synthesis of new DNA \citep{Tyson}; indeed, this is the first necessary step towards proliferation.

The structure of the paper is as follows. Section 2 provides the mathematical background for the modeling of CRNs for cell signaling. Section 3 contains the mathematical model describing LoF and GoF mutation processes. Section 4 is devoted to the application of the model to CRC cells. Our conclusions are offered in Section 5.

\section{Chemical reaction networks for cell signaling}
We consider a CRN consisting of $r$ reactions, denoted as $R_j$, $j=1, \dots , r$, that involve $n$ well-mixed reacting species, denoted as $A_i$, $i=1, \dots , n$. The network is modeled as a dynamical system with state vector $\mbf{x} =(x_1, \dots, x_n)^T \in \mathbb{R}^n_+$, where the upper $T$ denotes transposition, $\mathbb{R}_+$ is the set of non-negative real number, and the generic component $x_i$ is the molar concentration (nM) of the species $A_i$. According to this chemical interpretation, $\mbf{x}$ is also called concentration vector \citep{Yu}. We assume that the law of mass action holds: when two or more reactants are involved in a reaction, the reaction rate is proportional to the product of their concentrations. 
The resulting polynomial system of ODEs for the state variables is written as
 \begin{equation} \label{eq:dot_x}
\dot{\mbf{x}}= \mbf{S} \, \mbf{v}(\mbf{x},\mbf{k})~,
\end{equation}
where the superposed dot denotes the time derivative; $\mbf{k}=\left(k_1, \dots, k_r \right)^T \in \mathbb{R}^r_+$ stands for the set of rate constants; $\mbf{S}$ is the $\mathbb{R}^{n  \times r}$ constant stoichiometric matrix, $\mbf{v}(\mbf{x},\mbf{k}) \in \mathbb{R}^r_{+}$ is the vector of reaction fluxes. Here, system (\ref{eq:dot_x}) accounts for internal reaction and boundary fluxes \citep{Kschischo,Schilling} and hence is open \citep{Feinberg}. The matrix element $S_{ij}$ is the net number of molecules of the species $A_i$ that are produced whenever the reaction $R_j$ occurs. Thus the columns of $\mathbf{S}$ have been referred to as reaction vectors.
 
We are now interested in investigating the general properties of the solutions of the system (\ref{eq:dot_x}), and, in particular, in computing the corresponding asymptotically stable states. Our study is based on the analysis of the conservation laws and the stoichiometric compatibility classes of the system revised in the next two subsections.
 
 \subsection{Conservation laws and elemental species} \label{subsec:CL}
 
\begin{definition} Let $\mbf{x}(t)$ be a solution of the system of ODEs (\ref{eq:dot_x}). A constant vector $\boldsymbol{\gamma} \in \mathbb{N}^n \setminus  \left\{\mbf{0} \right\}$ is said to be a \textit{semi--positive conservation vector} if there exists $c \in \mathbb{R}_+$ such that 
\begin{equation}\label{eq:def_cl}
\bgamma^T \mbf{x}(t) = c \quad \forall t\ . 
\end{equation}
Moreover, the relation (\ref{eq:def_cl}) is called a \textit{semi--positive conservation law}, or equivalently \textit{moiety conservation law}.
\end{definition}

Since concentrations are expressed in nM, the biochemical interpretation of the conservation law is that the total
number of molecules involved in the species combination $\bgamma^T \mbf{x}$ remains constant during the evolution in
time of the network \citep{Shinar}. The constant total amount of available molecules is fixed by $\bgamma$ and the initial state $\mbf{x}_0$ through $c = \bgamma^T \mbf{x}_0$. Moreover, the concentrations of the species involved in the conservation law are bounded from above by the constant $c$ \citep{De_Martino,Haraldsdottir,Shinar,Schuster}. The following proposition relates the conservation vectors to the properties of the stoichiometric matrix.

\begin{proposition}\label{obs:cl_ker_St}
If $\bgamma \in \ker(\mbf{S}^T) \cap \mathbb{N}^n \setminus  \left\{\mbf{0} \right\}$, then $\bgamma$ is a conservation vector.
\end{proposition}
\begin{proof}
The thesis follows by observing that for any solution $\mbf{x}(t)$ of the system of ODEs (\ref{eq:dot_x})
 $$ \frac d{dt} (\bgamma^T \mbf{x}) = \bgamma^T  \dot{\mbf{x}} = \bgamma^T \mbf{S} \mbf{v}(\mbf{x}, \mbf{k})  = 0,$$ 
 where the last term is equal to zero as $\bgamma \in  \ker(\mbf{S}^T)$.
\end{proof}

We concentrate on conservation vectors $\bgamma \in  \ker(\mbf{S}^T)$. 
Proposition \ref{obs:cl_ker_St} implies that the set of all possible semi--positive conservation vectors defines a convex cone whose independent generators can be computed e.g. through the algorithm proposed by \cite{Schuster}. In the following we will denote with $\left\{ \bgamma_1, \dots, \bgamma_p \right\}$ a set of such generators and define the matrix $\mbf{N} \in \mathbb{R}^{p   \times n}$ as
\begin{equation} \label{eq:def_N}
  \mbf{N} = \begin{bmatrix} 
  \bgamma_1^T \\
  \vdots \\
  \bgamma_{p}^T \end{bmatrix} .
\end{equation}

Henceforth we will assume $p = n - \rank(\mbf{S})$ and thus the set $\left\{ \bgamma_1, \dots, \bgamma_p \right\}$ defines a basis for $\ker(\mbf{S}^T)$ 

Denote by $\mbf{x}_0$ the initial point of a trajectory.
Since $\mbf{N} \mbf{S} = 0$, it follows that 
\begin{equation} \label{eq:def_rho}
\mbf{N} \mbf{x}(t) =\mbf{N} \mbf{x}_0=: \mbf{c}
\end{equation}
is the constant vector in $\mathbb{R}^p$ whose components are the constants involved in the CLs. Thus the matrix $\mbf{N}$ determines $p$ linear, independent CLs; accordingly, the representative point $\mbf{x}(t)$ is constrained to move on the
affine subspace of $\mathbb{}{R}^n$ which is determined by equation
(\ref{eq:def_rho}), and hence is identified by $\mbf{N}$ and the initial data. Moreover, the linear affine subspace is the intersection of $p$ hyperplanes.\\

We conclude with a few remarks that will play a fundamental role in the analysis of the solutions of system (\ref{eq:dot_x}).
In general, there are chemical species that are not involved explicitly in the expressions of the CLs while the remaining species may belong to more than one CL.

\begin{definition}\label{def:elements}
We say that a CRN is \textit{elemented} if: (i) it admits a set of generators $\left\{\boldsymbol{\gamma}_1, \dots, \boldsymbol{\gamma}_p \right\}$ such that the matrix $\mbf{N}$ contains at least one minor equal to the identity matrix of order $p$, say $\mbf{I}_p$; (ii) each chemical species is involved in at least one CL, i.e. for each $i=1, \dots, n$ there exists $k \in \{1, \dots, p\}$ such that $\gamma_{ki} \neq 0$.
If only condition (i) is fulfilled, the CRN is \textit{weakly elemented}.

Borrowing the terminology of \cite{Shinar}, the species associated with the minor equal to the identity matrix will be called \textit{elemental} or \textit{basic species}.
\end{definition}

\begin{remark}
Definition \ref{def:elements} implies that each elemental species belongs to one, and only one conservation law. The idea is that elemental species consist of proteins in free form, which bind to other species involved in the same conservation law, in order to form the derived compounds or secondary species.

In general the set of elemental species of a network is not unique as it depends on the choice of the basis of the conservation vectors and, fixed a basis, multiple minors equal to the identity matrix may exist. 
\end{remark}

\begin{remark}
Given an elemented CRN, up to a change in the order of the components of $\mbf{x}$, the matrix $\mbf{N}$ may be decomposed as 
\begin{equation} \label{eq:M_dec}
\mbf{N} = \left[\mbf{I}_p, \mbf{N}_2 \right]
 \end{equation}
 with $\mbf{N}_2 \in \mathbb{R}^{p \times (n-p)}$. 
 Similarly, the concentration vector $\mbf{x}$ may be decomposed as 
 \begin{equation} \label{eq:x_dec}
\mbf{x} = \left(\begin{array}{c} \mbf{x}_1 \\ \mbf{x}_2 \end{array} \right)
\end{equation}
 with $\mbf{x}_1 \in \mathbb{R}^p$, and $\mbf{x}_2 \in \mathbb{R}^{n-p}$. In particular, $\mbf{x}_1 \in \mathbb{R}^p$ is formed by the elemental variables. Thus equation (\ref{eq:def_rho}) may be rewritten in the equivalent form
 \begin{equation} \label{eq:rho_2}
\mbf{x}_1  = \mbf{c} - \mbf{N}_2 \, \mbf{x}_2
\end{equation}
\end{remark}

Henceforth we assume the CRN to be weakly elemented, with elemental variables $x_1$ ... $x_p$, so that the matrix $\mbf{N}$ is described by the block decomposition (\ref{eq:M_dec}). 

According to (\ref{eq:rho_2}), the set of $p$ conservation equations is solved straightforwardly with respect to the elemental variables, thus yielding a parametric description of the affine space defined in (\ref{eq:def_rho}). Furthermore, substitution of the expression (\ref{eq:rho_2}) of $\mbf{x}_1$ into the system (\ref{eq:dot_x}) 
provides a reduced formulation of the original system of ODEs.
\begin{definition}
Consider an elemented CRN. A concentration vector $\mbf{x} \in \mathbb{R}^n$ is an \textit{ideal state} for the network iff only the elemental species have non-zero concentration, i.e., referring to equation (\ref{eq:x_dec}), $\mbf{x}_2 = 0$.
\end{definition}

\subsection{Stoichiometric compatibility classes}

Given the initial condition $\mbf{x}(0)=\mbf{x}_0$, consider the corresponding solution $\mbf{x}(t)$ of the system (\ref{eq:dot_x}), defined at least in the time domain $[0,T]$. Integration in time of both sides of (\ref{eq:dot_x}) leads to
 \begin{equation} \label{eq:int_x}
\mbf{x}(t) - \mbf{x}_0 = \mbf{S} \, \int_0^t \mbf{v}(\mbf{x}(\tau), \mbf{k})  \, d\tau\ ,  \quad t \in [0,T] .
\end{equation}  
This shows that $\mbf{x}(t)-\mbf{x}_0$ is a linear combination of the reaction vectors, with time-dependent coefficients. Therefore $\mbf{x}(t)-\mbf{x}_0$ belongs to a vector space of dimension equal to $\rank(\mbf{S})$ defined by the image of the stoichiometric matrix \citep{Feinberg,Feinberg_ARMA,Yu}. 
Accordingly, we provide the following definition of \textit{stoichiometric compatibility class} (SCC). 

\begin{definition}\label{def:scc}
Given a value $\mbf{x} \in \mathbb{R}^n_+$ of the state variable for system (\ref{eq:dot_x}), we define a \textit{stoichiometric compatibility class} (SCC) of $\mbf{x}$ the set
\begin{equation}
    \mathcal{SC}(\mbf{x}) = \left(\mbf{x} + \spn(\mbf{S}) \right) \cap \mathbb{R}^n_+ \ .
\end{equation}
\end{definition}

\begin{proposition}\label{prop:sc_and_cl}
Let $\left\{\bgamma_1, \dots, \bgamma_p \right\}$ be a set of generators of the the convex cone defined by the semi--positive conservation laws and let $\mbf{N}$ be the matrix defined as in equation (\ref{eq:M_dec}). If $p = n - \rank(\mbf{S})$ then for all $\mbf{x} \in \mathbb{R}^n_+$
\begin{align}
     \mathcal{SC}(\mbf{x}) & =  \left(\mbf{x} + \ker(\mbf{S}^T)^{\perp} \right) \cap \mathbb{R}^n_+ \label{prop:scc_1}\\
    & = \left\{\mbf{y} \in \mathbb{R}^n_+\ \text{s.t.}\ \mbf{N}\mbf{y} = \mbf{N}\mbf{x} \right\} . \label{prop:scc_2}
\end{align} 
\end{proposition}
\begin{proof}
Equation (\ref{prop:scc_1}) simply follows from Definition \ref{def:scc} by observing that $\spn(\mbf{S}) =  (\ker(\mbf{S}^T))^{\perp}$. Equation (\ref{prop:scc_2}) follows from the fact that, since $p = n - \rank(\mbf{S})$, the set $\left\{\bgamma_1, \dots, \bgamma_p \right\}$ is a basis for $\ker(\mbf{S}^T)$.
\end{proof}

\begin{definition}\label{def:gsa}
A CRN is said to satisfy the global stability condition if for every stoichiometric compatibility class there exists a unique globally asymptotically stable state $\mbf{x}_e$. 
\end{definition}

This means in particular that every trajectory with initial point on the given SCC tends asymptotically to the steady state $\mbf{x}_e$, which is also an equilibrium point. 
Details about asymptotic stability properties of CRNs can be found, e.g., in \citep{Chellaboina,Feinberg,Yu}.

\section{Mathematical model of Loss and Gain of Function mutations}
A mutation consists essentially in a permanent alteration in the nucleotide sequence of the genome of a cell. Mutations play a fundamental role in cancer evolution \citep{Stratton,Weinstein}. Here we are concerned with effects induced by mutations on species concentrations in the CRN. Specifically, we consider LoF and GoF mutations that are commonly observed in cancer cells \citep{Hochman, Lemieux, Levine_T, Levine}.
\\

LoF mutations result in reduction or abolishment of a protein function, which is simulated by a restriction on the value of the related density. The degree to which the function is lost can vary; for null mutations the function is completely lost and the concentration of the related molecules is supposed to vanish; for leaky mutations some function is conserved and the value of concentration is appropriately reduced \citep{Griffiths,Li}. In this study, we deal with null LoF mutations by referring directly to the concentrations of the mutated molecular species, and by assuming that they are set equal to zero.\\

GoF mutations are responsible for an enhanced activity of a specific protein, so that its effects become stronger \citep{Griffiths,Li}. In our framework a GoF is implemented by excluding from our CRN the reactions involved in the deactivation of the considered protein; this is achieved, by setting to zero the corresponding reaction rates.  \\

\subsection{Loss of Function}\label{sec:lof}
Consider a CRN described by system (\ref{eq:dot_x}) and consider a state $\mbf{x}$ of the system. A LoF mutation of the elemental species $A_j$ results in a projection of the state $\mbf{x}$ in a novel state where the concentrations of the $j-$th elemental species and of all its compounds are zero. Equivalently, also the total concentration $c_j$ available in the $j-$th conservation law is zero. This is modeled by applying the following operator to the state $\mbf{x}$.

\begin{definition}\label{def:lof_op}
Consider an elemented CRN and let $A_j$ be the $j-$th elemental species of the network. A LoF mutation of $A_j$ results in the operator $\mathcal{P}_{L_j} : \mathbb{R}^n \rightarrow \mathbb{R}^n$ that projects $\mbf{x}$ into a novel state $\mathcal{P}_{L_j}(\mbf{x}) = \left[\tilde{\mbf{x}}^T_1 ,  \tilde{\mbf{x}}^T_2\right]^T$ such that

\begin{displaymath}
\tilde{x}_{2,i} = 
\begin{cases}
0 & \text{if } \gamma_{ji} \neq 0 \\
x_{2, i} & \text{otherwise}
\end{cases}, \quad i=p+1, \dots, n
\end{displaymath}
and
\begin{displaymath}
\tilde{\mbf{x}}_1 = \left(\widetilde{\mbf{c}} - \mbf{N}_2 \tilde{\mbf{x}}_2 \right)
\end{displaymath}
where $\widetilde{\mbf{c}}$ is obtained by setting to 0 the $j-$th element of the vector $\mbf{c} = \mbf{N} \mbf{x}$.
\end{definition}

\begin{remark}
If $\mbf{x}$ is an ideal state for the CRN then $\mathcal{P}_{Lj}(\mbf{x})$ is obtained by setting to zero the concentration $x_j$ of the $j-$th elemental species. 
\end{remark}

\begin{remark}
$\mathcal{P}_{L_j}$ transforms SCCs into SCCs. Indeed, if the concentration vectors $\mbf{x}$ and $\mbf{y}$ belong to the same SCC, i. e.
\begin{displaymath}
\mbf{N} \mbf{x} = \mbf{N} \mbf{y} = \mbf{c}\ ,
\end{displaymath}
then $\mathcal{P}_{L_j} (\mbf{x})$ and  $\mathcal{P}_{L_j}(\mbf{y})$ still belong to the same SCC. More in details, 
\begin{displaymath}
\mbf{N} \mathcal{P}_{L_j}(\mbf{x}) = \mbf{N} \mathcal{P}_{L_j}(\mbf{y}) = \tilde{\mbf{c}} \ ,
\end{displaymath}
Where $\tilde{\mbf{c}}$ is equal to $ \mbf{c}$ except for the $j-$th element that is zero. 

Therefore the following theorem holds.
\end{remark}

\begin{theorem}\label{theo:Lof_res}
Consider an elemented CRN described by the system of ODEs (\ref{eq:dot_x}), and any states $\mbf{x}$, $\mbf{y}$ such that $\mathcal{SC}(\mbf{x}) = \mathcal{SC}(\mbf{y})$.  Then $\mathcal{SC}(\mathcal{P}_{L_j}(\mbf{x})) = \mathcal{SC}(\mathcal{P}_{L_j}(\mbf{y}))$. \\
If in addition the CRN satisfies the global stability condition then the trajectories starting from $\mathcal{P}_{L_j}(\mbf{x})$ and $\mathcal{P}_{L_j}(\mbf{y})$  lead to the same globally asymptotically stable state.
\end{theorem}

\begin{remark}
This holds in particular if $\mbf{x}$ and $\mbf{y}$ are replaced by the initial state $\mbf{x}_0$ and the corresponding steady state $\mbf{x}_e$, respectively. Thus the trajectories starting from $\mathcal{P}_{L_j}(\mbf{x}_0)$ and $\mathcal{P}_{L_j}(\mbf{x}_e)$  lead to the same mutated steady state.
\end{remark}

\subsection{Gain of Function}\label{sec:gof}
Consider a CRN described by system  (\ref{eq:dot_x}). In particular, let $\mbf{S}$ be the stoichiometric matrix of the system. A mutation resulting in the GoF of a given protein is implemented by removing from the network the reactions involved in the deactivation of such a protein. From a mathematical viewpoint this can be achieved by setting to zero the values of the corresponding rate constants, or equivalently by setting to zero the corresponding columns of $\mbf{S}$.  
\begin{definition}\label{def:gof_op}
Consider a CRN and let $\mbf{S}$ be the corresponding stoichiometric matrix. Given a set of reactions identified by the indices $H \subseteq \left\{1, \dots, r \right\}$ a GoF mutation results in the operator $\mathcal{G}_H : \mathbb{R}^{n \times r} \rightarrow \mathbb{R}^{n \times r}$ that projects $\mbf{S}$ into a novel stoichiometric matrix $\widetilde{\mbf{S}} = \mathcal{G}_H(\mbf{S})$ such that
\begin{displaymath}
\tilde{S}_{i, h} = 
\begin{cases}
0 & \text{if } h \in H \\
S_{i, h} & \text{otherwise}
\end{cases}, \quad i=1, \dots, n~,
\end{displaymath}
\end{definition}

\begin{remark}
$\mathcal{G}_H(\mbf{S})$ defines a new CRN where the chemical reactions in $H$ have been removed, while the set of chemical species, and thus the state space $\mathbb{R}^n$, are kept fixed.
\end{remark}

\begin{theorem}\label{theo:Gof_res}
If the set of reactions $H \subseteq \left\{1, \dots, r \right\}$ is such that 
\begin{equation}\label{eq:rank_cond}
    \rank \left( \mathcal{G}_H(\mbf{S}) \right) = \rank \left( \mbf{S} \right)
\end{equation}
then
\begin{equation}
    \ker \left( \mathcal{G}_H(\mbf{S})^T \right) = \ker \left( \mbf{S}^T \right)
\end{equation}
and thus in particular the stoichiometric matrices $ \mathcal{G}_H(\mbf{S})$ and $\mbf{S}$ define the same SCCs.
\end{theorem}
\begin{proof}
From Definition \ref{def:gof_op} it follows that $\ker \left( \mbf{S}^T \right) \subseteq \ker \left( \mathcal{G}_H(\mbf{S})^T \right).$
Indeed, if $\bgamma \in \ker \left( \mbf{S}^T \right)$ then 
\begin{displaymath}
\mbf{0} = \bgamma^T \mbf{S} = \left( \bgamma^T \mbf{S_1}, \dots, \bgamma^T  \mbf{S_r} \right)
\end{displaymath}
where $\mbf{S_j}$, $j=1, \dots, r$, denotes the $j-$th column of $\mbf{S}$. In particular $\bgamma^T \mbf{S}_j = 0$ for all $j = 1, \dots, r$, $j \not\in H$, that is $\bgamma \in  \ker \left( \mathcal{G}_H(\mbf{S})^T \right)$.

Additionally, equation (\ref{eq:rank_cond}) implies that $\ker \left( \mathcal{G}_H(\mbf{S})^T \right)$ and $\ker \left( \mbf{S}^T \right)$ have the same dimension and thus are equal.
\end{proof}

\begin{corollary}\label{cor:gof}
Consider a CRN described by the system of ODEs (\ref{eq:dot_x}) and satisfying the global stability condition of Definition \ref{def:gsa}. Let $\mbf{x}$ and $\mbf{y}$ such that $\mathcal{SC}(\mbf{x}) = \mathcal{SC}(\mbf{y})$. If $H \subseteq \left\{1, \dots, r \right\}$ is such that equation (\ref{eq:rank_cond}) holds and the CRN having stoichiometric matrix $\mathcal{G}_H(\mbf{S})$ satisfies the stability condition then the mutated trajectories starting from $\mbf{x}$ and $\mbf{y}$ lead to the same steady state.
\end{corollary}

\begin{remark}
This holds in particular if $\mbf{x}$ and $\mbf{y}$ are replaced by the initial state $\mbf{x}_0$ and the corresponding steady state $\mbf{x}_e$, respectively. When considering the CRN identified by $\mathcal{G}_H(\mbf{S})$ the trajectories starting from $\mbf{x}_0$ and $\mbf{x}_e$ lead to the same mutated steady state.
 \end{remark}

\subsection{Concatenation of mutation}

Many cancers arise by effect of a series of mutations accumulated in the cell over time. Here we generalize the results discussed in the previous sections to model the simultaneous action of multiple mutations on a given cell.

Consider for example a cell affected by two mutations resulting in the LoF of the elemental species $A_{j_1}$ and $A_{j_2}$. The most natural approach to quantify the combined effect of the two mutations is to start from a concentration vector $\mbf{x}_e$ modeling the (steady) state of a cell in physiological condition,  and then apply the procedure described in Section \ref{sec:lof} for each single mutation one after the other. Specifically, first we project $\mbf{x}_e$ through $\mathcal{P}_{L_{j_1}}$ and we compute the steady state of the trajectory started from $\mathcal{P}_{L_{j_1}}(\mbf{x}_e)$. Then we use $\mathcal{P}_{L_{j_2}}$ to project the novel, mutated steady state, and we compute the trajectory started from such a projection. The obtained trajectory will belong to the SCC
\begin{equation}\label{eq:aux_comb}
\left\{\mbf{y} \in \mathbb{R}_+^n \text{ s.t. } \mbf{N}\mbf{y} = \tilde{\mbf{c}} \right\}
\end{equation}
where $\tilde{\mbf{c}}$ has been obtained by setting to zero the $j_1-$th and the $j_2-$th element of $\mbf{c} := \mbf{N}\mbf{x}_e$.

The same result could have been obtained by reversing the order of the mutations. First we project $\mbf{x}_e$ through $\mathcal{P}_{j_2}$, then we apply $\mathcal{P}_{j_1}$ to the corresponding steady state. The resulting trajectory will still belong to the SCC described by equation (\ref{eq:aux_comb}) and thus will lead to the same steady state obtained in the previous case, provided that the CRN satisfy the global stability condition.

This result easily generalizes to an arbitrary set of combined GoF and LoF mutations through the following definition. 

\begin{definition}\label{def:comb_mut}
Consider an elemented CRN described by the system (\ref{eq:dot_x}) with stoichiometric matrix $\mbf{S}$. Consider a set of $\ell$ mutations, resulting in the LoF of the elemental species $A_{j_1}, \dots, A_{j_\ell}$, and a set of $q$ GoF mutations, resulting in the suppression of the sets of reactions $H_1, \dots, H_q \subseteq \left\{1, \dots, r \right\}$.

The combined effect of the considered mutations is quantified by computing the asymptotically stable state of the system of ODEs (\ref{eq:dot_x}) with stoichiometric matrix
\begin{equation}\label{def:comb_gof}
\mathcal{G}_{H_{q}} \circ \dots \circ \mathcal{G}_{H_1}(\mbf{S})
\end{equation}
and initial condition 
\begin{equation}\label{def:comb_lof}
\mbf{x}(0) = \mathcal{P}_{L_{j_{\ell}}} \circ \dots \circ \mathcal{P}_{L_{j_1}}(\mbf{x}_e)
\end{equation}
where $\circ$ denotes function composition and $\mbf{x}_e$ are the values of species concentration reached by the cell in  physiological condition.
\end{definition}

\begin{theorem}\label{theo:com_mut}
The steady state computed through the procedure described in Definition \ref{def:comb_mut} does not depend on the order of the composition in the equations (\ref{def:comb_gof}) and (\ref{def:comb_lof}) provided that 
\begin{equation}\label{theo:cond_comb}
    \rank \left( \mathcal{G}_{H_{q}} \circ \dots \circ \mathcal{G}_{H_1}(\mbf{S}) \right) = \rank(\mbf{S})
\end{equation}
and that all the involved CRNs satisfy the global stability condition of Definition \ref{def:gsa}.
\end{theorem}

\begin{proof}
The thesis follows by observing that the equations (\ref{def:comb_gof}) and (\ref{def:comb_lof}) do not depend on the order of the composition. Indeed the mutated stoichiometric matrix $\mathcal{G}_{H_{q}} \circ \dots \circ \mathcal{G}_{H_1}(\mbf{S})$ is defined by setting to zero the columns of $\mbf{S}$ corresponding to the reactions in $H_1 \cup \dots \cup H_q$ that clearly does not depend on the order in which the GoF mutations are considered. Additionally if condition (\ref{theo:cond_comb}) holds then $\mathcal{G}_{H_{q}} \circ \dots \circ \mathcal{G}_{H_1}(\mbf{S})$ and $\mbf{S}$ define the same SCCs as shown in Theorem \ref{theo:Gof_res}. The initial condition (\ref{def:comb_lof}) defines the unique SCC 
\begin{displaymath}
\mathcal{SC}(\mathcal{P}_{L_{j_{\ell}}} \circ \dots \circ \mathcal{P}_{L_{j_1}}(\mbf{x}_e)) = \left\{\mbf{x}\in \mathbb{R}^n_+ \text{ s.t. } \mbf{N}\mbf{x} = \tilde{\mbf{c}} \right\}
\end{displaymath}
where $\tilde{\mbf{c}}$ has been obtained from $\mbf{c} := \mbf{N}\mbf{x}_e$ by setting to zero the $j_i-$th element, for all $i = 1, \dots \ell$. A different order of the LoF mutations results in a different initial concentration vector on the same SCC. Therefore, according to Theorem \ref{theo:Lof_res} it leads to the same mutated asymptotically steady state.
\end{proof}

\begin{remark}
Very recent experimental studies on cancer cells have shown the importance of the order of mutations on the development of cancer diseases \citep{Levine}. This seems to be in contrast with the results stated in Theorem \ref{theo:com_mut}. However, that theorem follows from a model that does not account for the selection process induced on the mutated cell by the external environment. Instead, our model refers to the specific G1-S transition of a cancerous cell and it mimics the effects of accumulated mutations just for that phase.   
\end{remark}

\begin{remark}
Coherently to the previous Remark, the process described in Definition \ref{def:comb_mut} is equivalent to apply each considered mutation individually one after the other. If the hypotheses of Theorem \ref{theo:com_mut} hold, changing the order of the mutation affects the covered trajectory but will lead to the same final steady state. 
\end{remark}

\section{Application to the colorectal cancer cells}

\subsection{Generalities}

In this section we apply the previous procedures to the analysis of a CRN  which has been devised to provide a simplified description of how signals carried by a ligand from outside the cell are processed in order to determine the behavior of a CRC cell. The kinetic model applies to a healthy cell entering the G1-S phase of its development, when a synthesis of new DNA is carried out, as a first step in the process of cell division \citep{Tortolina}. Next we introduce a few mutations that transform the healthy cell into a cancer cell, and we analyze the related variations in the equilibrium concentrations.

We refer to the reaction network for CRC cells as CRC-CRN; 
its most relevant features are summarized below; the full list of chemical species together with their abbreviated names, the set of chemical reactions, the system of ODEs, the fixed values of the rate constants, and the initial concentrations of the elemental variables in an ideal physiologic state can be found in \citep{Tortolina}, Supplementary Materials.

CRC-CRN models a specific process of signal transfer in colorectal cells.
It is activated by external signals represented by  TGF$\beta$, WNT and EGF growth factors. Following activation, a cascade of chemical reactions proceeds. The relevant output of the system is given by the concentrations of a number of transcription factors (activators or repressors). CRC-CRN involves 8 constant, and 411 varying chemical species. The constant species model non-consumable chemicals: three constants describe the growth factors. Internal interactions are modeled by 339 reversible and 172 irreversible reactions, for a total of $r = 850$ reactions. Chemical kinetics is based on mass action law, so that 850 rate constants are considered. The state of CRC-CRN is described by $n = 419$ variables resulting from the concentrations of the 411 internal species, plus 8 additional variables, of null time derivative, accounting for the constant inflows. The state variables  satisfy a system of 419 polynomial ODEs as in equation (\ref{eq:dot_x}). In particular the monomials defining the reaction fluxes $\mbf{v}(\mbf{x}, \mbf{k})$ are quadratic, since the reactants considered depend at most on two chemical species. The rate constants enter the system as real and positive parameters.

\subsection{Conservation laws and elemental variables of the CRC-CRN}

We have obtained a basis of $\ker(\mbf{S}^T)$ consisting of $p=81$ semi-positive conservation vectors, providing an independent set of 81 CLs, all of them being regarded as moiety CLs. In particular, we point out that CRC-CRN satisfies the condition $p = n - \rank(\mbf{S})$. As expected, 8 CLs correspond to constancy of original non-consumable chemical species, so that they can be considered as trivial. The remaining 73 CLs describe effective properties of the system. 

It is seen by inspection that the CRC-CRN is weakly elemented, with 9 chemical species not involved in the conservation laws. Up to the 8 constant species, the list of elemental species coincides with that of the basic species of \citep{Tortolina}, which were defined as consisting of proteins in free form, which may bind to other species in order to form derived compounds or secondary species.

A simple, rather typical, example of conservation law within our CRC-CRN is 
\begin{equation} \label{eq:CL_44}
\text{NLK} + \text{NLKP} + \text{NLKP$_-$TCFLEF}+
 \text{NLKP$_-$Pase10} +  \text{TAKP$_-$TAB$_-$NLK}= c
 \end{equation}
where $c$ is a given constant;  here NLK denotes the concentration of the elemental species ``Nemo like kinase'' \citep{Tortolina}; NLPP is the phosphorylated form of NLK, while an expression like NLKP$_-$TCFLEF refers to the compound formed by NLKP and TCFLEF. Equation (\ref{eq:CL_44}) expresses conservation in time of the total number of molecules of NLK belonging to the compounds entering the combination in the left side, as discussed in \S \ref{subsec:CL}. In words, the NLK molecules are transferred between the metabolites involved in (\ref{eq:CL_44}), but are not synthesised, degraded or exchanged with the environment  \citep{De_Martino,Haraldsdottir}. The constant $c$, which is determined by the initial conditions, may be regarded as a counter of the conserved molecules of NLK. 

 \subsection{Global stability of CRC-CRN}
 
 There are a number of results available for equilibrium points and their stability properties \citep{Chellaboina,Domijan,Feinberg_ARMA,Yu}, but they cannot be applied straightforwardly to this CRC-CRN, essentially for two reasons. In fact, some of these results require very technical hypotheses that cannot be verified in the case of our ODEs system, due to its dimension and complexity. Further, other results cannot be applied in the case of open systems, like the one considered in this paper. Therefore, we will make use of numerical simulations to support the following conjecture.

\begin{conj}
CRC-CRN satisfies the global stability condition introduced in Definition \ref{def:gsa}.
\end{conj}
\textit{Numerical verification.}  We defined 5 different SCCs by randomly selecting the values of the total concentrations $\mbf{c}$. Specifically, each element $c_j$, $j = 1, \dots, p$,  has been drawn from a $\log_{10}$-uniform distribution on $[10^{-2}, 10^3]$. For each of the obtained SCC we generated 30 initial points  
\begin{displaymath}
\mbf{x}^{(k)}_0  \in \left\{\mbf{y} \in \mathbb{R}^n_+\ \text{s.t.}\ \mbf{N}\mbf{y} = \mbf{c} \right\}, \ k = 1, \dots, 30\, 
\end{displaymath}
as follows. 

\begin{itemize}
    \item First, we dealt with the species that do not belong to any conservation law. The value of their initial concentration were $\log_{10}-$uniformly drawn from the interval $[10^{-5}, 10^5]$.
    \item Then we considered all the other chemical species but the elemental species. After randomly permuting their order, for each species $i$:
\begin{enumerate}
    \item we randomly selected an initial concentration value below the upperbound imposed by the total concentrations available in the conservation laws involving it. More in details we set
    \begin{displaymath}
     x^{(k)}_{0, i} = u  \cdot \min_{j \in \Gamma(i)}{\frac{c_j}{\gamma_{ji}}} \ ,
     \end{displaymath}
     where $u$ was uniformly drawn from $[0, 1]$ and 
     \begin{displaymath}
    \Gamma(i) = \left\{ j \in \{1, \dots, p\}\ \text{s.t.}\ \gamma_{ji} \neq 0 \right\}.
    \end{displaymath}
    is the set of all the conservation laws involving the $i-$th species.
    \item we updated the total concentration available in each conservation law by removing the amount already filled by the $i-$th species, i.e.
    \begin{displaymath}
     c_j \leftarrow c_j - \gamma_{ji} \ x^{(k)}_{0, i} 
     \end{displaymath}
     for all $j \in \Gamma(i)$.
\end{enumerate}
\item Finally we considered the elemental species. By exploiting the fact that each elemental species belongs to only one conservation law, $j$, we set 
\begin{displaymath}
x^{(k)}_{0, i} = c_j 
\end{displaymath}
where $c_j$ is the value of the total concentration still available after the previous step.
\end{itemize}

For each of the 30 points generated with the described procedure, we used the matlab tool \texttt{ode15s} \citep{Shampine} to integrate the system of ODEs (\ref{eq:dot_x}) on the interval $[0, 2.5 \cdot 10^7]$, with initial condition $\mbf{x}(0) = \mbf{x}_0^{(k)}$. The value of the solution at the last time--point was considered as the corresponding asymptotic steady state $\mbf{x}_e^{(k)}$.

Since the initial values $\mbf{x}_0^{(k)}$ all belong to the same SCC they should lead to the same steady state. This was verified by computing for each species the coefficient of variation of the equilibrium values across the 30 trajectories. Namely, for each species we computed
\begin{equation}
    \varepsilon_e(i) = \frac{\sqrt{\frac{1}{29} \sum_{k=1}^{30} ( x_{e, i}^{(k)} - \mu_e(i) )^2 }}{\mu_e(i)} 
\end{equation}
where $\mu_e(i) = \frac{1}{30}\sum_{k=1}^{30} x_{e, i}^{(k)}$.

As a comparison, for each species we also computed the coefficient of variation $\varepsilon_0(i)$ across the initial values $x_{0, i}^{(k)}$, $k = 1, \dots, 30$.

Figure \ref{fig:uni} shows the averaged coefficient of variations obtained with the 5 considered SCCs. Fixed a SCC, the coefficient of variation across the initial values $\varepsilon_0$ is always around 2.5, while the coefficient of variation across the steady states $\varepsilon_e$ is at least one order of magnitude lower. Moreover, the latter shows higher differences across the SCCs. This is mainly due to the fact that depending on the SCC, few species may require a longer time to reach the asymptotically stable state and thus may show a higher variation when the system of ODEs is integrated in the fixed time-interval  $[0, 2.5 \cdot 10^7]$.

\begin{figure*}
  \includegraphics[width=0.9\textwidth]{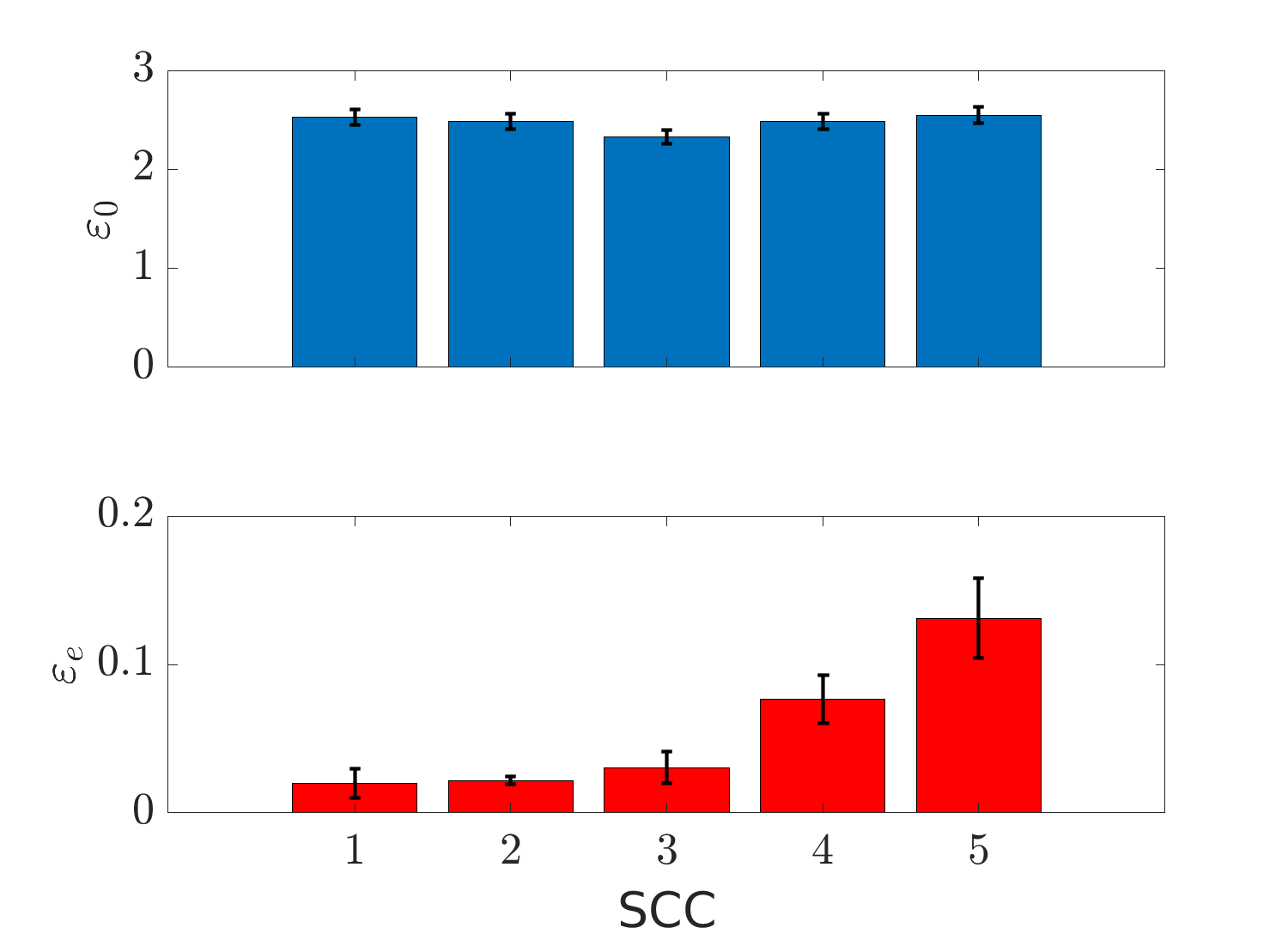}
\caption{Average and standard error of the mean over the non--constant species of the coefficient of variation across 30 different initial conditions (upper panel) and the corresponding asymptotically stable states (lower panel). Each bar corresponds to the results obtained with a different SCC. Note the different scale on the $y$ axes.}
\label{fig:uni}
\end{figure*}

\subsection{LoF mutation of \text{TBRII}}

We investigated the impact on CRC-CNR of a mutation resulting in the LoF of the elemental species \text{TBRII}. 

First, we computed the value of the concentration vector modeling the physiological steady state of the cell prior to mitosis. To this end, we integrated the system of ODEs (\ref{eq:dot_x}) on the interval $[0, 2.5 \cdot 10^7]$, with initial condition $\mbf{x}(0) = \mbf{x}_0$, where $\mbf{x}_0$ is the ideal state defined by setting the initial concentrations of the elemental species as in \cite{Tortolina}. The value of the steady state in the physiological cell was defined as the value of the solution at the last time--point.

By using Definition \ref{def:lof_op}, we then defined the operator $\mathcal{P}_{L_j}$ associated to the LoF of \text{TBRII}. The steady state of the mutated cell was computed by integrating the system of ODEs (\ref{eq:dot_x}) with two different initial conditions, namely $\mbf{x}(0) = \mathcal{P}_{L_j}(\mbf{x}_0)$ and $\mbf{x}(0) = \mathcal{P}_{L_j}(\mbf{x}_e)$. As in the previous step, in both cases we solved the system on the interval $[0, 2.5 \cdot 10^7]$ and we defined as steady state the value of the solution computed at the last time--point.

According to Theorem \ref{theo:Lof_res} the two trajectories, starting from $\mathcal{P}_{L_j}(\mbf{x}_0)$ and $\mathcal{P}_{L_j}(\mbf{x}_e)$, should lead to the same steady state. This result was numerically verified by computing for both the trajectories the relative difference
\begin{equation}
    \delta_i = \frac{x^m_{e, i} - x_{e, i}}{x_{e, i}}\ , \quad i = 1, \dots, n
\end{equation}
where $\mbf{x}_e^m$ is the steady state in the mutated cell.

As shown in Figure \ref{fig:lof_eq}, the results obtained with the two different initial conditions coincide and thus Theorem  \ref{theo:Lof_res} is verified. Moreover, Figure \ref{fig:lof_eq} shows the effect of the LoF of \text{TBRII} on the concentrations of all the chemical species. Indeed, $\delta_i < 0$ means that the value of the concentration of the $i-$th species in the mutated cell is lower than in the healthy cell. On the contrary $\delta_i > 0$ means that the amount of the $i-$th species is higher in the mutated cell.

On the other hand, different values of the initial conditions lead to different trajectories. In Figure \ref{fig:lof_speed} we compared the time required by the two trajectories to reach a stable state. To this end for each time point $t$ we computed 
\begin{equation}
    ||\dot{\mbf{x}}(t) ||_{\infty} = ||\mbf{S} \mbf{v}(\mbf{x}(t), \mbf{k}) ||_{\infty} \ .
\end{equation}
Since $\mbf{x}_e$ is already a steady state for the CRC-CRN, when computing the projected value $\mathcal{P}_{L_j}(\mbf{x}_e)$ the species not affected by the LoF of \text{TBRII} maintain their stable values. As a consequence, the trajectory starting from $\mathcal{P}_{L_j}(\mbf{x}_e)$ requires a lower number of iterations to reach a value of $||\dot{\mbf{x}}(t) ||_{\infty} $ closer to 0.

\begin{figure*}
  \includegraphics[width=0.9\textwidth]{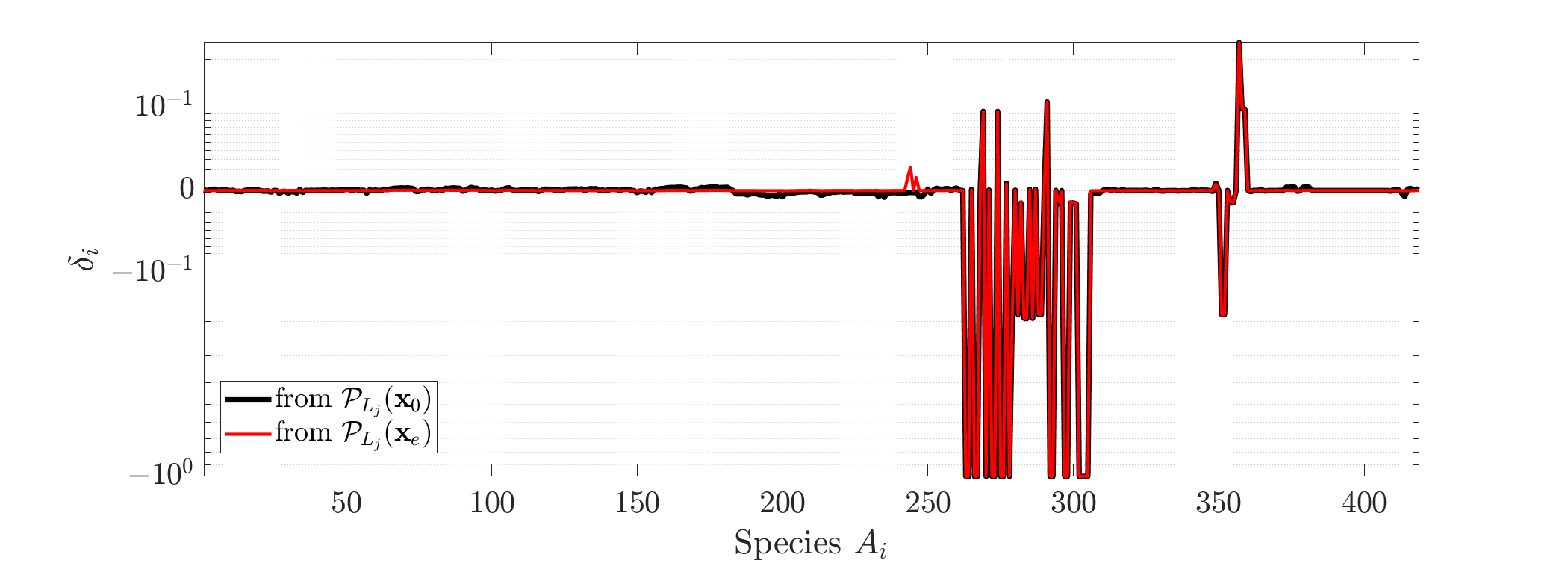}
\caption{Value of the relative difference $\delta_i$ between the steady states in the physiological cell and in the cell affected by LoF mutation of \text{TBRII}. Black and red lines are obtained when the system of ODEs for the mutated cell is solved with initial condition $\mbf{x}(0) = \mathcal{P}_{L_j}(\mbf{x}_0)$ and $\mbf{x}(0) = \mathcal{P}_{L_j}(\mbf{x}_e)$, respectively.}
\label{fig:lof_eq}
\end{figure*}

\begin{figure*}
  \includegraphics[width=0.9\textwidth]{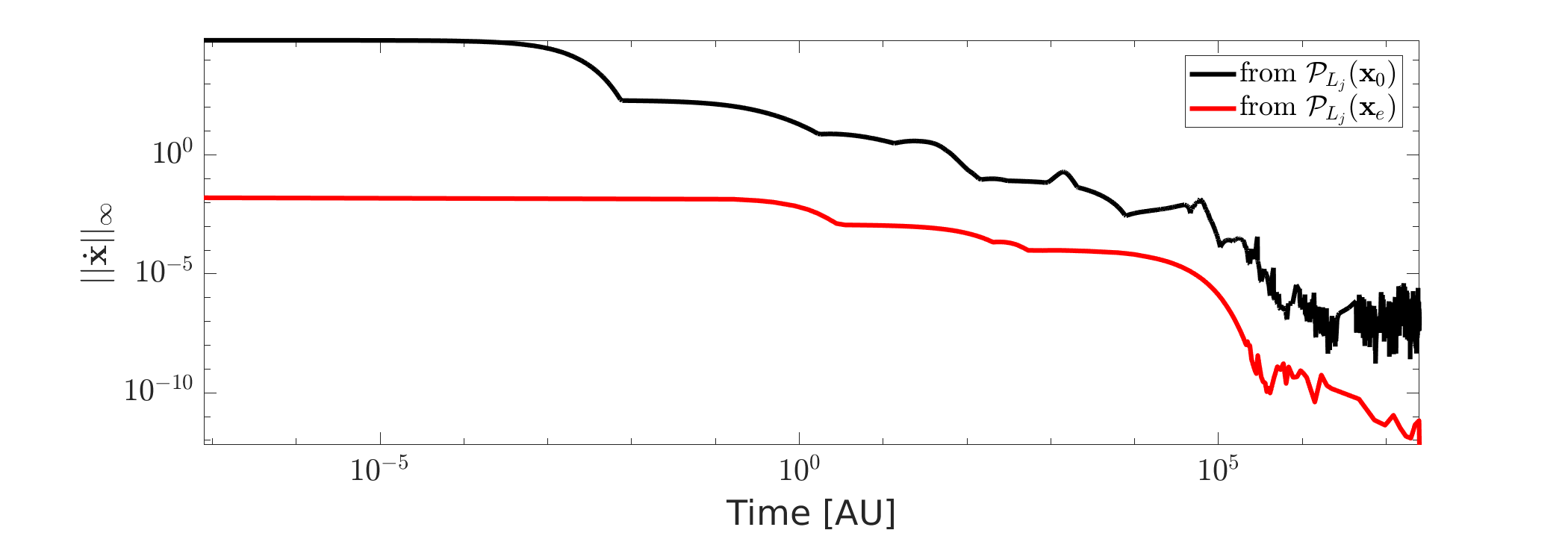}
\caption{Infinity norm of the derivative $\dot{\mbf{x}}$ of the concentration vector as function of time. Black and red lines show the results obtained by solving the ODEs system (\ref{eq:dot_x}) with initial condition $\mbf{x}(0) = \mathcal{P}_{L_j}(\mbf{x}_0)$ and $\mbf{x}(0) = \mathcal{P}_{L_j}(\mbf{x}_e)$, respectively, $\mathcal{P}_{L_j}$ being the operator associate to the LoF mutation of \text{TBRII}}
\label{fig:lof_speed}
\end{figure*}

\subsection{GoF mutation of \text{BRAF}}
We investigated the impact on the CRC-CNR of a mutation resulting in the GoF of the elemental species \text{BRAF}.

To this end, following Definition \ref{def:gof_op}, we defined the operator $\mathcal{G}_H$, where the $H$ is defined to include all the reactions involved in the deactivation of \text{BRAF}$^*$. We recall that \text{BRAF}$^*$ is the activated form of \text{BRAF}, consisting in the phosphorylation of a specific amino acid. Thus it is assumed that the mutated form of \text{BRAF} is still subject to activation in \text{BRAF}$^*$, while inactivation of \text{BRAF}$^*$ is blocked.

As described in \cite{Tortolina}, Supplementary materials, the deactivation of \text{BRAF}$^*$  is regulated by the phosphatase \text{Pase1} through the following set of reactions:
\[  \text{BRAF}^* + \text{Pase1}  \; \overset{k_{1f}}{\underset{k_{1r}}\rightleftarrows}    \; \text{BRAF}^* _-\text{Pase1}
\]  
\[  \text{BRAF}^* _-\text{Pase1}  \;  \overset{k_2}\rightarrow \;   \text{BRAF} + \text{Pase1} \]
where a simplified notation has been adopted for the rate constants.
To block such a deactivation process while respecting the condition described by equation (\ref{eq:rank_cond}) in Theorem \ref{theo:Gof_res} we removed the two forward reactions 
\[  \text{BRAF}^* + \text{Pase1}  \; \overset{k_{1f}}\rightarrow    \; \text{BRAF}^* _-\text{Pase1}
 \;  \overset{k_2}\rightarrow \;   \text{BRAF} + \text{Pase1} , \]
that is we defined a novel stoichiometric matrix $\mathcal{G}_H(\mbf{S})$ were the corresponding column of $\mbf{S}$ have been set equal to 0.

Similarly to what we have done in the previous section, we computed the asymptotically stable states of the trajectories obtained solving the system of ODEs defined by the stoichiometric matrix $\mathcal{G}_H(\mbf{S})$ with two different initial condition, namely $\mbf{x}(0) = \mbf{x}_0$ and $\mbf{x}(0) = \mbf{x}_e$,  $\mbf{x}_0$ and $\mbf{x}_e$ being the ideal and steady state for the physiological cell.

Figure \ref{fig:gof_eq} shows that the two trajectories lead to the same steady state as stated in Corollary \ref{cor:gof}. Additionally, as done in the previous section, in Figure \ref{fig:gof_speed} we compare the time required by the two trajectories to reach the stable state. According to Definition \ref{def:gof_op}, the GoF of \text{BRAF} is implemented by setting to zero a set of columns of the stoichiometric matrix $\mathbf{S}$, and thus modifying the system of ODEs associated to the CRC-CRN. Although $\mbf{x}_e$ was a stable state for the system of ODEs modeling the cell in physiological condition, for most of the species the corresponding concentration $x_{e, i}$ is no longer an equilibrium value for the new system. For this reason the two trajectories, starting from $\mbf{x}_e$ and $\mbf{x}_0$, show similar behaviors of $||\dot{\mbf{x}}(t) ||_{\infty} $ as function of the time $t$ in which the solution of the system is computed.


\begin{figure*}
  \includegraphics[width=0.9\textwidth]{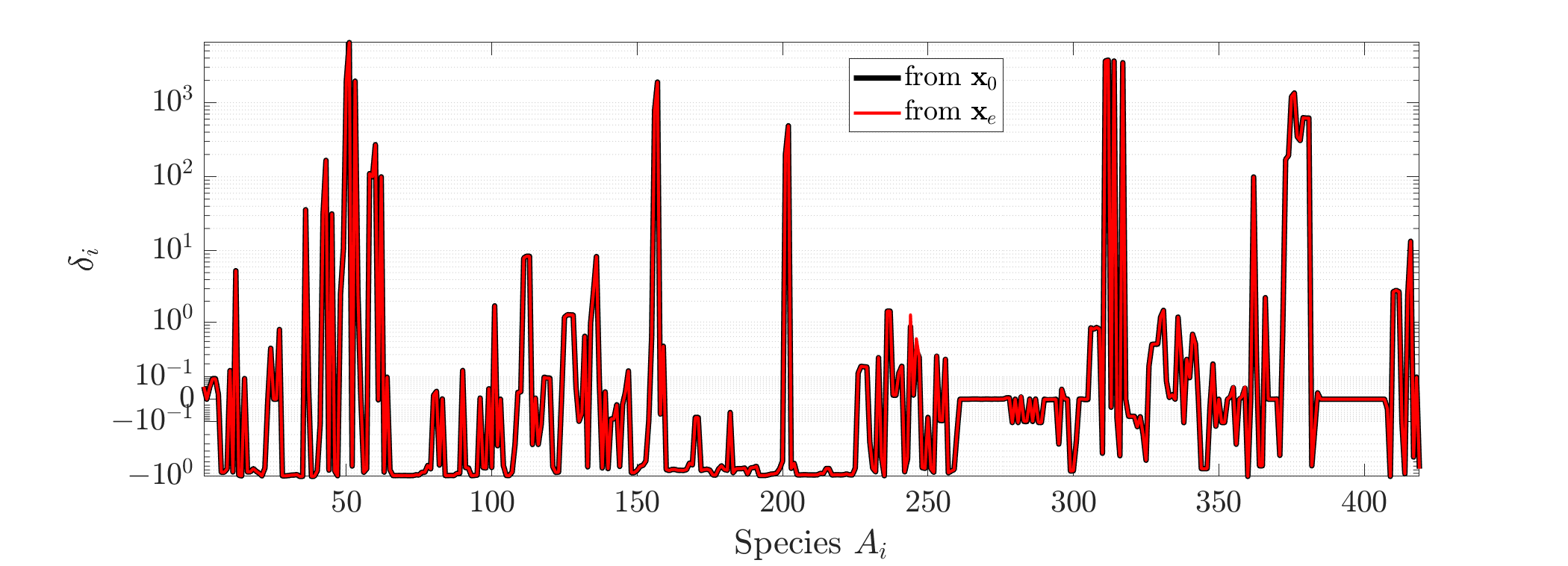}
\caption{Value of the relative difference $\delta_i$ between the steady states in the physiological cell and in the cell affected by GoF mutation of \text{BRAF}. Black and red lines are obtained when the mutated system of ODEs defined by the stoichiometric matrix $\mathcal{G}_H(\mbf{S})$ is solved with initial condition $\mbf{x}(0) = \mbf{x}_0$ and $\mbf{x}(0) = \mbf{x}_e$, respectively.}
\label{fig:gof_eq}
\end{figure*}

\begin{figure*}
  \includegraphics[width=0.9\textwidth]{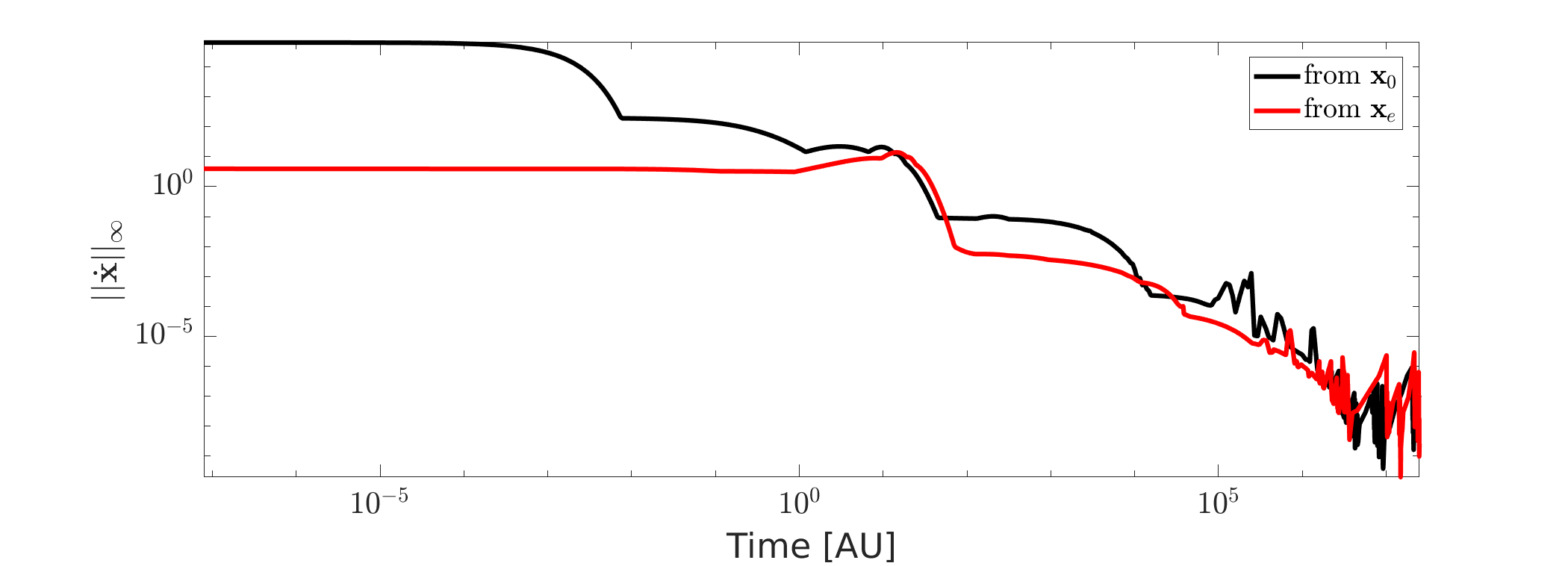}
\caption{Infinity norm of the derivative $\dot{\mbf{x}}$ of the concentration vector as function of time when the stoichiometric associated to GoF of \text{BRAF} is employed in the ODEs system (\ref{eq:dot_x}). Black and red lines show the results obtained by setting $\mbf{x}(0) = \mbf{x}_0$ and $\mbf{x}(0) = \mbf{x}_e$, respectively.}
\label{fig:gof_speed}
\end{figure*}

\section{Discussion and future directions}
In this paper we have considered a system of ODEs modeling a CRN, with emphasis on the mutual relationships between moiety CLs and stoichiometric surfaces. CLs have led to the characterization of weakly elemented 
CRN and the definition of ideal states, which are crucial in the treatment of mutations. Confining attention to systems satisfying the  global  stability  condition, we have considered two classes of mutations, LoF and GoF, often found in cancer cells. Mutations have been modeled as projection operators and basic properties of mutated networks have been analyzed, aiming at the determination of mutated equilibrium states as asymptotic steady limits of trajectories. The new model for mutations allows for a simple treatment of their combinations, showing that the resulting equilibrium is independent of the order.

As an application of the previous results we have investigated a system of ODEs describing the G1-S phase of a CRC cell. Due to the huge number of variables involved, the analysis has been based on numerical simulations. 
We have found 81 independent, linear, moiety CLs; they have been applied to support the conjecture that CRC-CRN satisfies the global stability condition. It has also been found that the elemental chemical species coincide with the basic species defined by \cite{Tortolina} on biochemical grounds. Also, a mutation by LoF, and another mutation by GoF have been examined in detail.

We are aware that this computational approach to CRN is devised to capture few, although decisive, aspects of the G1-S transition point of a cell, as well as the modifications of the network induced by cancer mutations. We think that the model can provide a guide to future experimental research based on the interpretation of results of simulations, particularly in the case of the development and optimization of targeted therapies agains already mutated cells. For example, if the simulation predicts a significant increase of the mutated concentration of a specific metabolite with respect to its physiological value, then that metabolite can be regarded as a potential drug target. More in general, an analysis of the mutated profile of the simulated CRN in the G1-S phase should provide hints about convenient choices of targets for available drugs, together with a framework for the simulation of their consequences \citep{Fachetti,Torres}. 

Possible development of our model may involve different directions. For example, we could extend it to examine alterations of mRNA induced by changes from physiologic to mutated equilibrium \citep{Castagnino}. Further, in the model we have not considered possible dependence of the literature data on temperature, pH, or other conditions. Therefore, the scheme is not capable of describing the individual response of a fixed cell; rather, it attempts at simulating the behavior of a homogeneous set of cells; from this viewpoint,  the solution of the system of ODEs may be interpreted as providing an average dynamic of the intra-cellular answer. Also, the scheme should be enriched in order to account for the impact of the extra-cellular micro-environment, which implies to modify the model in order to account for growth-induced solid stresses and pressure, nutrients and oxygen supply, and blood perfusion \citep{Caviglia,Jones,Markert}.

As a final comment, we remark that the rate constants in the CRC-CRN considered in this paper are assumed as known and estimated from the scientific literature \citep{Tortolina}. A more reliable determination of these constants can be obtained by solving a non-linear ill-posed problem  \citep{bertero2006inverse, Engl} in which the input data are provided by {\it{ad hoc}} conceived biochemical experiments. The setup of such experiments and the realization of an optimization method for the numerical solution of this inverse problem will be part of future research.

\section*{Acknowledgments}
MP has been partially supported by Gruppo Nazionale per il Calcolo Scientifico. The authors kindly acknowledge Prof. Silvio Parodi, Prof. Silvia Ravera and Dr. Mara Scussolini for their useful comments and feedback.

\bibliographystyle{spbasic}      
\bibliography{biblio}   

\end{document}